\documentclass[format=acmsmall,review=false,nonacm]{acmart}
\usepackage{booktabs} %
\usepackage[ruled]{algorithm2e} %

\SetAlFnt{\small}
\SetAlCapFnt{\small}
\SetAlCapNameFnt{\small}
\SetAlCapHSkip{0pt}
\IncMargin{-\parindent}

\setcitestyle{acmnumeric}
\setcitestyle{numbers,sort&compress}

\title{An Entropy Potential for Type-Composition Games}

\author{Morteza Alimi}
\email{morteza.alimi@uni-a.de}
\orcid{0000-0003-3732-2903}
\affiliation{
    \institution{University of Augsburg}
    \city{Augsburg}
    \country{Germany}
}

\author{Merlin de la Haye}
\email{merlin.delahaye@hpi.de}
\orcid{0009-0000-9517-4609}
\affiliation{
    \institution{Hasso Plattner Institute, University of Potsdam}
    \city{Potsdam}
    \country{Germany}
}

\author{Pascal Lenzner}
\email{pascal.lenzner@uni-a.de}
\orcid{0000-0002-3010-1019}
\affiliation{
    \institution{University of Augsburg}
    \city{Augsburg}
    \country{Germany}
}

\author{\v{S}imon Schierreich}
\email{schiesim@fit.cvut.cz}
\orcid{0000-0001-8901-1942}
\affiliation{%
	\institution{AGH University}
	\city{Krakow}
	\country{Poland}
}
\affiliation{%
	\institution{Czech Technical University in Prague}
	\city{Prague}
	\country{Czechia}
}
    
\author{Alexander Skopalik}
\email{a.skopalik@utwente.nl}
\orcid{0000-0002-4950-8708}
\affiliation{
    \institution{University of Twente}
    \city{Twente}
    \country{The Netherlands}
}

\author{Marcus Wunderlich}
\email{marcus.wunderlich@uni-a.de}
\orcid{0009-0008-0519-5421}
\affiliation{
    \institution{University of Augsburg}
    \city{Augsburg}
    \country{Germany}
}

\renewcommand{\shortauthors}{M. Alimi, M. de la Haye, P. Lenzner, Š. Schierreich, A. Skopalik, and M. Wunderlich}


\begin{abstract}
    Potential functions are a key tool in theoretical computer science with applications ranging from the runtime analysis of algorithms and data structures, through the analysis of the expected behavior of random processes and search heuristics, to proving the existence of equilibrium states in strategic games. Typically, proofs that employ potential functions are short, elegant, and easy to verify, yet very powerful. Moreover, potential functions are essential ingredients for constructive proofs, in particular in algorithmic game theory. There, a key question is the existence of equilibrium states, but the most powerful theorem in the field -- Nash's theorem -- is unfortunately non-constructive. For many strategic games, potential functions come to the rescue by enabling constructive proofs that sometimes even yield efficient algorithms for finding equilibria.
    
    We add to this by providing a novel class of entropy-inspired log-multinomial potential functions for natural game-theoretic settings where rational agents of different types strategically choose actions to maximize their utility. In particular, we consider utility functions that are based on the fraction of same- and other-type agents taking the same action. We demonstrate the versatility of the new potential function class by presenting simple equilibrium existence proofs for two recent game-theoretic models, for which only involved technical proofs were previously known. Even better, the new potential function class yields efficient algorithms for constructing equilibria for much more general models. Thereby, we positively resolve several open problems.
\end{abstract}

\usepackage{mathtools}
\usepackage{amsthm}
\usepackage{thmtools}
\usepackage{nicefrac}

\usepackage{cleveref}

\usepackage{xspace}
\usepackage{titlecaps}
\Addlcwords{and with the a an}

\usepackage{enumerate}

\usepackage{todonotes}
\usepackage{tikz}
\usetikzlibrary{arrows.meta,positioning}

\newif\ifshortJournalName
\shortJournalNamefalse
\newif\ifshortConferenceName
\shortConferenceNamefalse

\newcommand{\gameName}[1]{\emph{\titlecap{#1}}\xspace}
\def\BnM{\textsc{BMG}}

\newcommand{\agents}{N}
\newcommand{\numAgents}{n}

\newcommand{\numTypes}{q}
\newcommand{\type}{T}
\newcommand{\agentType}[1]{t(#1)}

\newcommand{\actions}{A}
\newcommand{\allActions}{\mathcal{A}}
\newcommand{\profile}{\mathbf{s}}
\newcommand{\action}{s}
\newcommand{\agentsCount}{n}

\newcommand{\util}{u}

\newcommand{\bakers}{B}
\newcommand{\millers}{M}
\newcommand{\locations}{L}

\newcommand{\Oh}[1]{{\mathcal{O}\mathopen{}\left(#1\right)}}

\definecolor{myBlue}{HTML}{4664aa}
\definecolor{myRed}{HTML}{a22223}
\definecolor{myLightRed}{HTML}{c37070}

\newtheorem{remark}{Remark}

\begin{document}

\maketitle

\section{Introduction}\label{sec:intro}
In the analysis of game-theoretic settings, usually the first and most important question that researchers investigate is whether the game in question admits an equilibrium state. Such states are vectors of strategies selected by each player such that no player unilaterally wants to deviate to another strategy. Thus, equilibria are stable states of the game-theoretic system and thus serve as behavioral predictions. This has numerous applications in economics and the social sciences. The most commonly used equilibrium concept is the \emph{Nash equilibrium}, which is guaranteed to exist due to Nash's theorem~\cite{nash1950equilibrium}. However, the downside of Nash's proof is that it cannot be turned into an algorithm for finding equilibria, since it is non-constructive. Moreover, the existence guarantee only holds for the case where the players' strategies can randomize between different actions. 

For many applications, ranging from network creation, through coalition formation, job scheduling, and resource selection, to location decisions in a competitive market, it is more natural to assume that players select only a single action as their strategy, e.g., a specific location for opening a facility. Equilibria consisting only of such strategies are called \emph{pure Nash equilibria} (PNE), and it is well-known that PNE might not exist, i.e., Nash's theorem does not apply, and it is computationally hard to find~\cite{FabrikantPT2004}. However, there is a large class of strategic games, called \emph{potential games}, that does admit PNEs. The key feature of these games is that a \emph{generalized ordinal potential function}~\cite{MondererS1996} exists, which maps states of the game to values such that the potential value improves if a player can improve individually by changing their strategy. Then, the state that maximizes the potential value must be a PNE. Thus, the existence of PNEs can be proven by finding a suitable potential function. Moreover, this can be turned into a constructive algorithm by iteratively performing improving strategy changes of the players, e.g., in round-robin order, until a (local) maximum of the potential function is reached. This approach is called \emph{improving response dynamics} (IRD).

The idea of using such an algorithm to find PNEs works for all games in which every sequence of improving strategy changes for the players must be finite. This is called the \emph{finite improvement property} (FIP). The seminal work of Monderer and Shapley~\cite{MondererS1996} established that this is actually a characterization: every game having the FIP admits a generalized ordinal potential function. Thus, the constructive IRD algorithm is guaranteed to find a PNE iff such a potential function exists.  

One of the most prominent examples of a potential function in algorithmic game theory is the Rosenthal potential~\cite{Rosenthal1973}, which works for the broad class of congestion games. There, players strategically select resources, and their utility depends on the number of players on each shared resource. Most importantly, this dependence on the number of players, instead of the specific set of players on each resource, is critical for Rosenthal-style potentials. Thus, such approaches work for settings where the players are anonymous. 

In this paper, we consider \gameName{Type-Composition Games} which are close in spirit to congestion games, since also there players select resources (called \emph{actions} in our model) and their utility depends on the action choices of all players. However, we stray
from the beaten path by dropping the anonymity assumption, i.e., we consider players of different types, and the utility of a player depends not on the number but on the distribution of types of players selecting the same action. In particular, we focus on the fraction of same-type and other-type players.

For this broad setting of \gameName{Type-Composition Games}, we provide a generalized ordinal potential function and apply it to two previously studied games, which are special cases of our model. Thereby, we obtain clean and technically simple proofs of the existence of PNEs. We also derive efficient IRD algorithms to find them. Our results supersede previous results on these games, as we provide a much simpler PNE existence proof for cases where PNE existence was already known, and we prove PNE existence for cases where it was still open. Moreover, we obtain even stronger results by extending the existence of PNE to generalizations of these games. Last but not least, our approach directly yields that, in all these cases, PNEs can be found in polynomial time. %

\subsection{Related Work}

Potential functions are a fundamental tool across theoretical computer science. The \emph{potential method} used in \emph{amortized analysis}~\cite{Tarjan1985} has been successfully applied to provide tight analysis of self-adjusting data structures such as splay trees~\cite{SleatorT1985}, Fibonacci heaps~\cite{FredmanT1987}, and union-find~\cite{Tarjan1975}. Next, the \emph{drift analysis} for random processes such as evolutionary algorithms~\cite{Hajek1982,HeY2001,DoerrJW12,Witt13,Lengler2020} uses a bound on the expected per-step difference of a potential function to derive a bound on the expected convergence time. Another prominent example is \emph{algorithmic game theory}, where \gameName{potential games}~\cite{MondererS1996,FabrikantPT2004} are one of the few general classes for which pure Nash equilibria are guaranteed to exist and can be reached via improving response dynamics. This class of games includes many important models such as \gameName{congestion games}~\cite{Rosenthal1973,HarksKM2011,MILCHTAICH1996111}, \gameName{selfish routing}~\cite{RoughgardenT2002} \gameName{cost-sharing games}~\cite{AnshelevichDKTWR2008,BiloFM2020}, or \gameName{load balancing games}~\cite{EvenDarKM2007,FeldottoLMS2019}.

The \gameName{Bakers and Millers game} was introduced by Aziz \emph{et al.}~\cite{AzizBBHOP2019} as a special case of \gameName{(fractional) hedonic games}~\cite{DrezeG1980,BogomolnaiaJ2002,AzizBBHOP2019,BiloFFMM2018}. There, each agent is of one of two types, either a baker or a miller, and prefers to be in a coalition with as large a fraction of agents of the opposite type as possible. More recently, Krogmann \emph{et al.}~\cite{KrogmannLS2025} generalized this model by restricting each baker to an agent-specific set of allowed locations, while millers can still choose any location. They show that pure Nash equilibria continue to exist and can be computed in polynomial time. Importantly, they left whether a potential function exists for the game as an open problem.
Also related are the so-called \gameName{hedonic diversity games}~\cite{BredereckEI2019,BoehmerE2020}, where agents are also partitioned into two types, but this time, their preferences are arbitrary linear orders over the possible compositions of each coalition in terms of the present types. Ganian \emph{et al.}~\cite{GanianHKSS2023} generalized the model to multiple types and provided an extensive parameterized complexity analysis for several stability notions. 

A different generalization is the model of \gameName{strategic resource selection with homophilic agents}~\cite{GadeaHarderKLS2023}, where agents of two types pick resources on which at least a $\tau$-fraction of other agents share their type. The authors consider two variants, with impact-blind and impact-aware players, respectively. The latter players know the exact numbers of type distribution on the resources, while in the former case, only the fractions, but not the exact numbers, are known. For the impact-blind case, the social welfare serves as a potential function, and thus PNEs exist and can be found in $\Oh{n^5}$ many steps, where $n$ is the number of players. Interestingly, the impact-aware setting behaves very differently. There, the authors show that the game has the FIP if $\tau \leq \nicefrac{1}{2}$, proven via a vector-based function that is almost a potential. However, it is an open problem if PNE exist, and if the FIP holds for $\tau > \nicefrac{1}{2}$ too.  

The \gameName{bakers and millers game} can also be understood as the simultaneous move variant of two-stage facility location games~\cite{twostageKrogmannLMS21,twostageKrogmannLS23,twostageKrogmannLSUV24}, where agents of two types, facilities and clients, select locations in a given network in two stages: first, facility agents select a location, then clients select which facilities to patronize. There, some variants do not admit equilibria, while for others, equilibrium existence was shown via sophisticated potential function arguments. 

Inspired by Schelling's famous segregation model~\cite{Schelling1969,Schelling1971}, several game-theoretic refinements have been introduced~\cite{ChauhanLM2018,Echzell0LMPSSS2019,AgarwalEGISV2021}. The mainstream of research focuses on selfish agents located at the vertices of a graph and unilaterally seeking to deviate to vertices with a high fraction of same-type neighbors~\cite{BrandtIKK2012,BiloBLM2022b,BullingerSV2021,BiloBLM2022a,KanellopoulosKV2023,KreiselBFN2024,DeligkasEG2024}. However, recently, also models with single-peaked utilities~\cite{BiloBLM2022a,singlepeakedFriedrichLMS23}, diversity-~\cite{NarayananSV2025} and variety-seeking~\cite{NarayananOTV2025} utilities, and variants with continuous types~\cite{BiloBDLMS2023,continuousHLSW26} have been considered. Notably, in some cases, potential functions based on social welfare, the number of same-type neighbors, or vector-based functions using utilities have been identified. In some cases, the existence of a potential function depends on the graph topology and the value of the tolerance threshold $\tau$~\cite{Echzell0LMPSSS2019,BiloBLM2022b}. However, Agarwal \emph{et al.}~\cite{AgarwalEGISV2021} showed that there are variants that do not admit PNEs. Finding potential functions of many of these variants was repeatedly posed as an open question. 
Closely related is also the model of \gameName{refugee housing}~\cite{KnopS2023,Schierreich2024,Schierreich2023}, where there are two types of agents -- the inhabitants and the refugees -- and the goal is to allocate the refugees to the topology with the inhabitants pre-allocated so that the preferences of all agents are respected.

\subsection{Our Contribution}

We introduce \gameName{Type-Composition Games}, where players of different types select actions, and their utility depends on the fraction of same-type and other-type players that also take their chosen action. This covers a broad range of natural settings, since it contains the previously studied \gameName{bakers and millers games}  (BMG)~\cite{AzizBBHOP2019,KrogmannLS2025} and \gameName{resource selection games with homophilic agents} (RSG)~\cite{GadeaHarderKLS2023}, which both have multiple applications, as special cases. 

As our main result, we provide a novel type of potential function that utilizes log-multinomials and is inspired by the concept of entropy. To the best of our knowledge, such a potential has not yet been applied to game-theoretic settings (and we also do not know of any similar approaches in algorithm analysis), and thus this contribution might be of independent interest as it may be applied to other settings with similar features.  

We use the log-multinomial potential function to prove the FIP for the \gameName{type-composition game} with $q$ types, and we show that for any such game, improving response dynamics are guaranteed to converge in $\Oh{\numAgents^3\cdot |\allActions|\log \numTypes}$ many steps, where $\allActions$ is the set of possible actions of the players. Thus, for games with a polynomial-size description, we obtain an efficient algorithm for finding a PNE. 

We apply our general theorem to a known class of games. In particular, we show that the BMG is a special case of the \gameName{type-composition game} with two types. Thus, our potential function yields a much simpler PNE existence proof, compared to the original proof in~\cite{KrogmannLS2025}. Moreover, we obtain a much simpler, efficient algorithm for finding a PNE. Additionally, although the BMG in~\cite{KrogmannLS2025} has only one-sided location restrictions, our constructive proof also holds when both bakers and millers have location restrictions. This resolves an explicitly stated open problem by the authors.

Finally, we remark that our novel potential function yields a simpler PNE existence proof for 
\gameName{Resource Selection Games with Homophilic Agents}~\cite{GadeaHarderKLS2023}. This also allows resolving open questions regarding existence in the general case and yields a polynomial-time algorithm to compute a PNE.

\section{Type-Composition Games}\label{sec:diversityGames}

We introduce the class of \gameName{Type-Composition Games}, a framework that generalizes existing models such as \gameName{Bakers and Millers games}~\cite{AzizBBHOP2019,KrogmannLS2025} and \gameName{Resource Selection games with homophilic agents}~\cite{GadeaHarderKLS2023}. 

\paragraph{Players and Types.}
Let $\agents = \{1, \ldots, \numAgents\}$ be the set of \emph{players}, partitioned into $\numTypes$ \emph{types} $\type_1\, \cup \cdots  \cup \type_\numTypes$, that is, $\bigcup_{i=1}^\numTypes \type_i = \agents$ and $\type_i \cap \type_j = \emptyset$ for every $i\not= j$.
For convenience, we set $\agentType{i}=j$ if player $i$ belongs to type $\type_j$.

\paragraph{Actions and Strategies.}
Each player $i \in \agents$ selects an \emph{action} from a set $\actions_i$; we denote the set of all possible actions by $\allActions = \bigcup_{i \in \agents} \actions_i$. A \emph{strategy profile} is a vector $\profile = (\action_1, \ldots, \action_\numAgents) \in \actions_1 \times \cdots \times \actions_\numAgents$. For a given profile $\profile$, let $\agentsCount_\action^j(\profile)$ denote the number of players of type $\type_j$ choosing action $\action \in \allActions$. The total number of players choosing action $\action$ is then given by $\agentsCount_\action(\profile) = \sum_{j=1}^\numTypes \agentsCount_\action^j(\profile)$.%

\paragraph{Utilities and Equilibria.}
The utility of player $i$ in profile $\profile$ is denoted by $\util_i(\profile)$, and we will define the specifics of it later when discussing several variants. As is standard, we denote by $(\action^\prime_i,\profile_{-i})$ the strategy profile resulting from $\profile$ when player $i$ replaces their action $\action_i$ with $\action'_i$.
An action $\action_i \in \actions_i$ is a \emph{best response} of player $i\in\agents$ to a partial profile $\profile_{-i}$ if:
\[  \forall \action^\prime_i \in \actions_i\colon \quad
    \util_i(\action_i, \profile_{-i}) \geq \util_i(\action^\prime_i, \profile_{-i}).
\]
A strategy profile $\profile = (\action_1, \ldots, \action_n)$ is a \emph{pure Nash equilibrium (PNE)} if, for every player $i \in \agents$, action $\action_i$ is a best response to $\profile_{-i}$. Formally, profile~$\profile$ is a PNE if:
\[
\text{for all } i \in \agents \text{ and for all } \action^\prime_i \in \actions_i\colon \util_i(\action_i, \profile_{-i}) \geq \util_i(\action^\prime_i, \profile_{-i}).
\]

\subsection{Type-Composition Functions}\label{sec:utilities}

In \gameName{Type-Composition Games}, players seek to choose their actions to maximize or minimize the ratio of the number of same-type players. There are two natural ways to define this ratio:

The \emph{type-balance} of a player $i$ in profile $\profile$ will serve as the utility\footnote{For notational uniformity, we use the term utility for both maximizing and minimizing games, although for the latter the term cost is the more appropriate term. }, denoted by $\util_i(\profile)$. It is defined as the number of different-typed players choosing the same action divided by the number of same-typed players
, i.e., 
\[
    \util_i(\profile) = 
        \frac{\sum_{j\ne \agentType{i}} \agentsCount^j_{\action_i}(\profile)}
        {\agentsCount^{\agentType{i}}_{\action_i}(\profile)} 
        = 
        \frac{\agentsCount_{\action_i}(\profile) - \agentsCount^{\agentType{i}}_{\action_i}(\profile)}{\agentsCount^{\agentType{i}}_{\action_i}(\profile)}.
\]
Alternatively, one can define the type-balance as the total number of players who choose that action divided by the number of players of the same type who choose the same action, i.e.,%
\[
    \util^\prime_i(\profile) = 
        \frac{\agentsCount_{\action_i}(\profile)}{\agentsCount^{\agentType{i}}_{\action_i}(\profile)}.   
\]
Note that all used fractions are always well-defined, as at least the player whose type-balance/utility is calculated chooses the considered action.

\subsection{Potential Functions}
Given a strategy profile $\profile$, an \emph{improving move} of a maximizing\footnote{The case of minimizing players is analogous: improving moves decrease $u_i$ and $\phi$.} player~$i \in \agents$ is a change from action~$\action_i$ to some other action $\action_i^\prime$ such that
\[
    \util_i(\action_i,\profile_{-i}) < \util_i(\action_i^\prime,\profile_{-i})\,.
\]
A sequence of strategy profiles $\profile^1,\ldots,\profile^k$ is called an \emph{improvement path} if, for every $t \in \{1,\ldots,k-1\}$, there is a player~$i \in \agents$ such that strategy profile $\profile^{t+1}$ emerges from the previous strategy profile $\profile^{t}$ by an improving move by player~$i$. A game has the \emph{finite improvement property (FIP)}, if all improvement paths have finite length. Monderer and Shapley~\cite{MondererS1996} proved that having the FIP is equivalent to a game admitting a \emph{generalized ordinal potential function}. This is a function $\phi$ that maps strategy profiles to real numbers such that every improving move of some player~$i$ from action $\action_i$ to action~$\action_i'$ increases the value of $\phi$, i.e.,
\[
    \util_i(\action_i,\profile_{-i}) < \util_i(\action_i^\prime,\profile_{-i}) \quad\implies\quad \phi(\action_i,\profile_{-i}) < \phi(\action_i^\prime,\profile_{-i})
\]
holds for all strategy profiles $\profile$ and all improving moves from $\action_i$ to $\action_i'$. %

\section{The Novel Entropy-Inspired Log-Multinomial Potential}\label{sec:potential}%

As our main contribution, we show in \Cref{thm:potentialFunction} that all \gameName{type-composition games} have the finite improvement property by proving that
\[
    \Phi(\profile)
    \coloneqq
    \sum_{\action \in \allActions} \log\binom{\agentsCount_\action}{\agentsCount_\action^1,\ldots,\agentsCount_\action^\numTypes}
    =
    \sum_{\action\in \allActions}
    \log\!\left(
        \frac{\agentsCount_\action!}{\prod_{j=1}^{\numTypes} \agentsCount^j_\action(\profile)!}
    \right)
\]
is a generalized ordinal potential function, where $\binom{\agentsCount_\action}{\agentsCount_\action^1,\ldots,\agentsCount_\action^\numTypes}$ is the multinomial coefficient. 

\begin{remark}[Entropy interpretation]
    By Stirling's approximation, we have $\log\binom{\agentsCount_\action}{\agentsCount_\action^1, \ldots, \agentsCount_\action^\numTypes} = \agentsCount_\action \cdot H(\agentsCount_\action^1/\agentsCount_\action, \ldots, \agentsCount_\action^\numTypes/\agentsCount_\action) + \Oh{\numTypes \log \agentsCount_\action}$, where $H$ denotes the Shannon entropy. Thus, up to lower-order terms, $\Phi$ is a weighted aggregation of the entropy of the type distribution at each action.%
\end{remark}

Now, we prove our main theorem.

\begin{theorem}\label{thm:potentialFunction}
    $\Phi$
is a generalized ordinal potential function for all \gameName{type-composition games}.
\end{theorem}

We now show that $\Phi$ is a generalized ordinal potential function for all games in which players maximize or minimize the type-balance functions $\util_i$ or $\util^\prime_i$. To that end, the following \Cref{lemma:potentialDiffIsUtilityDiff} proves that a type-balance-increasing or type-balance-decreasing unilateral deviation of a player leads to a proportional increase or decrease of the potential function values, respectively.

\begin{lemma}\label{lemma:potentialDiffIsUtilityDiff}
    Suppose a player $i\in\agents$ changes from its action $\action_i$ to an action $\action^\prime_i$. Then
    \[
        \Phi(\action^\prime_i,\profile_{-i})-\Phi(\profile)
        =
        \log\!\left(
            \frac{1+\util_i(\action^\prime_i,\profile_{-i})}{1+\util_i(\profile)} 
        \right)
           =
        \log\!\left(
            \frac{\util'_i(\action^\prime_i,\profile_{-i})}{\util^\prime_i(\profile)} 
        \right) 
           =
        \log(\util'_i(\action^\prime_i,\profile_{-i})) - \log(\util^\prime_i(\profile)).
    \]
\end{lemma}
\begin{proof}
    Let $x=\action_i$ and $y=\action^\prime_i$ be the action before and after the deviation, respectively, and let $t=\agentType{i}$ be the type of player $i$. We write as shorthand
    \[
        \agentsCount_x=\agentsCount_x(\profile) \qquad\text{and}\qquad \agentsCount_y=\agentsCount_y(\profile),
    \]
    for the number of players using action $x$ and action $y$, respectively, and we write
    \[
        \agentsCount_x^t=\agentsCount_x^t(\profile)\qquad\text{and}\qquad \agentsCount_y^t=\agentsCount_y^t(\profile)
    \]
    for the number of players of type $t$ using the respective action.
    
    Clearly, after the deviation, $\agentsCount_x^t$ decreases by one and $\agentsCount_y^t$ increases by one, and correspondingly, $\agentsCount_x$ decreases and $\agentsCount_y$ increases by one. 
    When considering the difference in potential between $\Phi(\action^\prime_i,\profile_{-i})$ and $\Phi(\profile)$, only terms belonging to action $x$ and~$y$ change, thus the other terms vanish in the difference.
    What remains is the following, which can be simplified by first combining the logarithms and then canceling the factorials in the numerator and denominator:   
    \begin{align*}
        \Phi(\action^\prime_i,\profile_{-i})-\Phi(\profile)
        &=
        \log\!\left(
            \frac{(\agentsCount_y+1)!}{(\agentsCount_y^{t}+1)!\prod_{j\neq t}\agentsCount_{y}^j!}
            \cdot
            \frac{(\agentsCount_x-1)!}{(\agentsCount_x^t-1)!\prod_{j\neq t}\agentsCount_{x}^j!}
            \right)-\log\!\left(
            \frac{\agentsCount_y!}{\prod_j \agentsCount_{y}^j!}
            \cdot
            \frac{\agentsCount_x!}{\prod_j \agentsCount_{x}^j!}
        \right) \\
        &=
        \log\!\left(
            \frac{(\agentsCount_y+1)!}{(\agentsCount_y^{t}+1)!\prod_{j\neq t}\agentsCount_{y}^j!}
            \cdot
            \frac{(\agentsCount_x-1)!}{(\agentsCount_x^t-1)!\prod_{j\neq t}\agentsCount_{x}^j!}
            \cdot
            \frac{\prod_j \agentsCount_{y}^j!}{\agentsCount_y!}
            \cdot
            \frac{\prod_j \agentsCount_{x}^j!}{\agentsCount_x!}
        \right) \\
        &=
        \log\!\left(
            \frac{\agentsCount_y+1}{\agentsCount_y^{t}+1}
            \cdot
            \frac{\agentsCount_{x}^t}{n_x}
        \right).
    \end{align*}
    Next, we can expand both sub-fractions and insert the definition of player $i$'s utility, i.e., their type-balance, before and after the deviation, obtaining
    \[
        \frac{\agentsCount_y+1}{\agentsCount_y^{t}+1}
        =
        \frac{\agentsCount_y+1-(\agentsCount_y^t+1)+(\agentsCount_y^t+1)}{\agentsCount_y^{t}+1}
        =
        1+\util_i(\action^\prime_i,\profile_{-i})
    \]
    and
    \[
        \frac{\agentsCount_x}{\agentsCount_x^t}
        =
        \frac{\agentsCount_x-\agentsCount_x^t+\agentsCount_x^t}{\agentsCount_x^t}
        =
        1+\util_i(\profile).
    \]
    Furthermore, we have
    \[
        \frac{\agentsCount_y+1}{\agentsCount_y^{t}+1}
        =
        \util'_i(\action^\prime_i,\profile_{-i})
        \qquad\text{and}\qquad
        \frac{\agentsCount_x}{\agentsCount_x^t}
        =
        \util'_i(\profile).
    \]
    By substituting, we obtain, as desired:
    \[
        \Phi(\action^\prime_i,\profile_{-i})-\Phi(\profile)
        =
        \log\!\left(
            \frac{1+\util_i(\action^\prime_i,\profile_{-i})}{1+\util_i(\profile)}
        \right)   =
        \log\!\left(
            \frac{\util'_i(\action^\prime_i,\profile_{-i})}{\util'_i(\profile)}
        \right)
        =
        \log(\util'_i(\action^\prime_i,\profile_{-i})) - \log(\util^\prime_i(\profile)).\qedhere
    \]
\end{proof}

In our next result, we show not only that $\Phi$ is a potential function (so PNEs can be found by following improving moves), but also that improvement sequences converge in polynomial time.

\begin{theorem}
    \label{thm:time}
    In \gameName{Type-Composition Games}, every sequence of improving moves terminates after at most $\Oh{\numAgents^3\log \numTypes}$ many steps.
\end{theorem}
\begin{proof}
    We observe that for any action $\action\in\allActions$ it holds that
    \[
        \frac{\agentsCount_\action(\profile)!}{\prod_{i=1}^{\numTypes} \agentsCount_\action^i(\profile)!} \leq \frac{\agentsCount_\action(\profile)!}{(\agentsCount_{\action}(\profile)/\numTypes)!^\numTypes}\le 
        \sum_{k_1+\cdots+k_\numTypes = \agentsCount_{\action}(\profile)} \frac{\agentsCount_{\action}(\profile)!}{k_1!\cdots k_\numTypes!} = 
        \numTypes^{\agentsCount_\action(\profile)},
    \]
    and therefore, we can bound the potential of every profile $\profile$ by
    \begin{align}\label{eq:potentialbound}
        0\le \Phi(\profile) =  \sum_{\action\in \allActions}
        \log\!\left(
            \frac{\agentsCount_\action(\profile)!}{\prod_{j=1}^{\numTypes} \agentsCount^j_\action(\profile)!}
        \right) \le \sum_{\action \in \allActions} \agentsCount_\action(\profile)\log \numTypes=\numAgents\log \numTypes.
    \end{align}
    
    To bound the maximal number of steps, we consider the change of the potential of an improving move.  \Cref{lemma:potentialDiffIsUtilityDiff} gives
    \[
        \Phi(\profile')-\Phi(\profile)
        =
        \log\!\left(
            \frac{1+\util_i(\profile')}{1+\util_i(\profile)}
        \right)
    \qquad
    \text{and}
    \qquad
        \Phi(\profile')-\Phi(\profile)
        =
        \log\!\left(
            \frac{\util'_i(\profile')}{u'_i(\profile)}
        \right).
    \] 
    
    For ease of notation let $\util_i(\profile) = \frac{x}{y}$ and $\util_i(\profile^\prime)= \frac{x'}{y'}$ for some integers $1\leq x,y,x',y' \leq \numAgents$. Then
    \[
        \frac{1+\util_i(\profile')}{1+\util_i(\profile)}
        = \frac{1+ \frac{x'}{y'}}{1+ \frac{x}{y}}=\frac{(x'+y')\cdot y}{(x+y)\cdot y'}
        \geq
        1+\frac1{2\numAgents^2},
    \]
    where the last inequality follows from the fact that all variables are integers at most $\numAgents$ and that the term is strictly larger than $1$.
    Similarly,
    \[
        \frac{\util^\prime_i(\profile')}{\util^\prime_i(\profile)}
        = \frac{ \frac{x'}{y'}}{\frac{x}{y}}=\frac{x'y}{xy'}
        \ge 
        1+\frac1{\numAgents^2}.
    \]
    Therefore, the change of potential can be lower bounded by
    \begin{align}\label{eq:potchange}
        \Phi(\profile')-\Phi(\profile)
        \ge
        \log\!\left(1+\frac1{\numAgents^2}\right)
        \ge
        \frac{1}{\numAgents^2+1}.   
    \end{align}
    Thus, by combining \eqref{eq:potentialbound} and \eqref{eq:potchange}, every sequence of improving moves terminates after at most
    $%
        \Oh{\numAgents^3\log \numTypes}
    $ %
    many moves, finishing the proof.
\end{proof}

Finding an improving player can be easily done in $\Oh{\numTypes\cdot |\allActions| + \numAgents\cdot |\allActions|}$ time by calculating the fractions for each action and each type once, and then checking for each player if it can improve at each action. Therefore, we obtain the following.

\begin{corollary}
    Computing a pure Nash equilibrium for a given \gameName{Type-Composition Game} can be done in $\Oh{\numAgents^4\cdot|\allActions| \log \numTypes}$ time.
\end{corollary}

\section{Applications}

We now show how our main theorem simplifies and generalizes known results.

\subsection{Bakers and Millers}

\begin{figure}[t]
    \centering
    \includegraphics[width=0.9\linewidth]{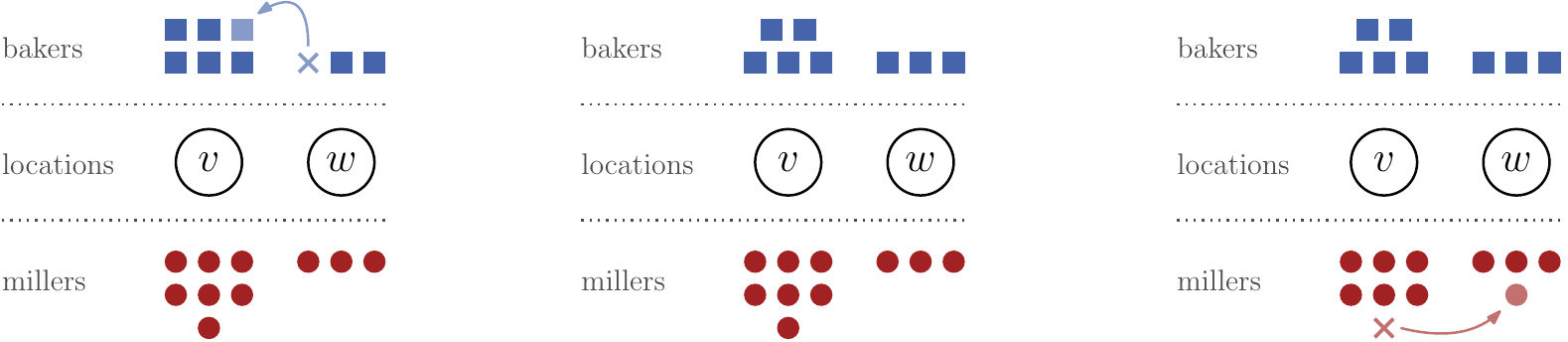}
    \caption{Example instance of the \gameName{bakers and millers game} with two locations $v$ and $w$, eight bakers (blue squares) and ten millers (red circles). Based on the strategy profile in the middle, a baker in location $w$ can improve by moving to location $v$ (see left, utility $1 \leadsto \frac{7}{6}$), and a miller in location $v$ can improve by moving to location $w$ (see right, utility $\frac{5}{7} \leadsto \frac{3}{4}$).}
    \label{fig:bm-example}
\end{figure}

The \gameName{bakers and millers game} (\BnM{}) was introduced by Aziz \emph{et al.}~\cite{AzizBBHOP2019} and later studied by Krogmann \emph{et al.}~\cite{KrogmannLS2025}
to model strategic location choice by customers and sellers, referred to as the bakers and millers. This is an example of a diversity-seeking game; for example, millers seek to choose locations with many bakers (to sell flour to) but few other millers (to minimize competition). See \Cref{fig:bm-example} for an example of an instance of the \gameName{bakers and millers game}. Now, we introduce the game more formally, starting with the original definition from~\cite{KrogmannLS2025}.

\paragraph{The Bakers and Millers Game with One-Sided Location Restrictions.}
There is a set of \emph{locations} $\locations$ and a set of players or agents $\agents = \bakers \cup \millers$, where players of $\bakers$ are called \emph{bakers} and players of $\millers$ are called \emph{millers}. Each baker $b\in\bakers$ has a set of feasible locations $\locations_b \subseteq \locations$ they may choose from. A miller can choose any location in $\locations$.
A player's strategy is the choice of a feasible location.
The payoff of a baker $b\in\bakers$ at a location $\ell\in\locations$ is the number of millers that also chose $\ell$ divided by the number of bakers at $\ell$. The payoff for millers is analogously defined to be the inverse.

\paragraph{The Bakers and Millers Game with Two-Sided Location Restrictions.}
A natural generalization of the model is to also allow location restrictions for the millers, i.e., each miller $m\in\millers$ has a set of feasible locations $\locations_m \subseteq \locations$, and their strategy is restricted to these locations. We call the resulting model the \gameName{bakers and millers game with two-sided Location Restrictions}.\\

We show that even the more general version of the \gameName{bakers and millers game} is a special type of \gameName{type-composition game}, and therefore, we directly obtain all desirable properties of the latter class.

\begin{theorem}\label{th:bam}
    The \gameName{bakers and millers game with two-sided Location Restrictions} is a \gameName{type-composition game}.
\end{theorem}
\begin{proof}
    We set $\type_1 = \bakers$ and $\type_2 = \millers$. Next, for every baker $b\in\bakers$, we set $\actions_b = \locations_b$. Similarly, for each miller $m \in \millers$, we set $\actions_m = \locations_m$. This yields an equivalent \gameName{type-composition game} with $q=2$ types. It is straightforward that the utility function $\util_i$ exactly reflects the utility function of the \gameName{bakers and millers game with two-sided Location Restrictions}.
\end{proof}

Krogmann \emph{et al.}~\cite{KrogmannLS2025} proved the existence of pure Nash equilibria in the basic model by giving an explicit algorithm to compute one. This non-trivial algorithm crucially depends on the property that the players on one side, i.e., the millers, can be assigned to arbitrary locations. The existence of a potential function was left as an intricate open question. By \Cref{th:bam} we can answer this question positively, even while generalizing their model to the \gameName{bakers and millers game with two-sided Location Restrictions}, i.e., with the possible restriction of the feasible locations for both~sides.

\begin{corollary}
    The \gameName{bakers and millers game with two-sided Location Restrictions} has the finite improvement property, and hence admits a pure Nash equilibrium.
\end{corollary}

Moreover, \Cref{thm:time} implies a polynomial-time procedure to compute pure Nash equilibria even for the most general version of the \gameName{bakers and millers game}. However, we note that for the basic version, the algorithm of Krogmann \emph{et al.}~\cite{KrogmannLS2025} is faster, running in $\Oh{\numAgents^3}$ time.

\subsection{Resource Selection with Homophilic Agents}

The \gameName{Resource Selection Game with Homophilic Agents} (RSG)~\cite{GadeaHarderKLS2023}, is also a special case of our \gameName{type-composition game}  with two types.
Due to space constraints, we only sketch our results for these games. They essentially correspond to games in which agents seek to maximize the inverse of $u'_i$. However, the utility is capped at a given parameter $\tau$.
Gadea Harder \emph{et al.}~\cite{GadeaHarderKLS2023} prove PNE existence only for the special case where $\tau \leq \nicefrac{1}{2}$. This proof is non-trivial as it uses a function that almost behaves like a potential, but some strategy changes of the players do not change the potential value. The authors then show that such neutral changes can only occur finitely often. 

Our approach yields a much simpler PNE existence proof and, most importantly, it works for all values of the parameter $\tau$. 
This answers a conjecture from the paper positively. Moreover, it was also open whether PNEs can be found efficiently, which we also resolved positively. Last but not least, our novel potential function can be applied to the RSG with more than two types, i.e., we can show PNE existence and efficient construction for this more general model. 

\section{Conclusion}
We introduce a class of potential functions that, to the best of our knowledge, has not been used in Algorithmic Game Theory (nor in other fields) before. This novel potential has a natural interpretation stemming from information theory entropy. 

Equipped with our new powerful tool, we derive simple and elegant PNE existence proofs for two previously studied types of games. Moreover, the versatility of the novel potential allows us to positively resolve several open problems. This shows that new and stronger results can be obtained by truly simple techniques like generalized ordinal potential functions. 

\begin{acks}
    This research was (partially) funded by the HPI Research School on Foundations of AI (FAI), from the \grantsponsor{ERC}{European Research Council}{} (ERC) under the European Union's Horizon 2020 research and innovation programme (grant agreement No \grantnum{ERC}{101002854}), and by the \grantsponsor{EU}{European Union} co-funded project Robotics and Advanced Industrial Production (reg. no. \grantnum{EU}{CZ.02.01.01/00/22\_008/0004590}).

    \vspace{0.1cm}
    \begin{center}
        \includegraphics[width=2.25cm]{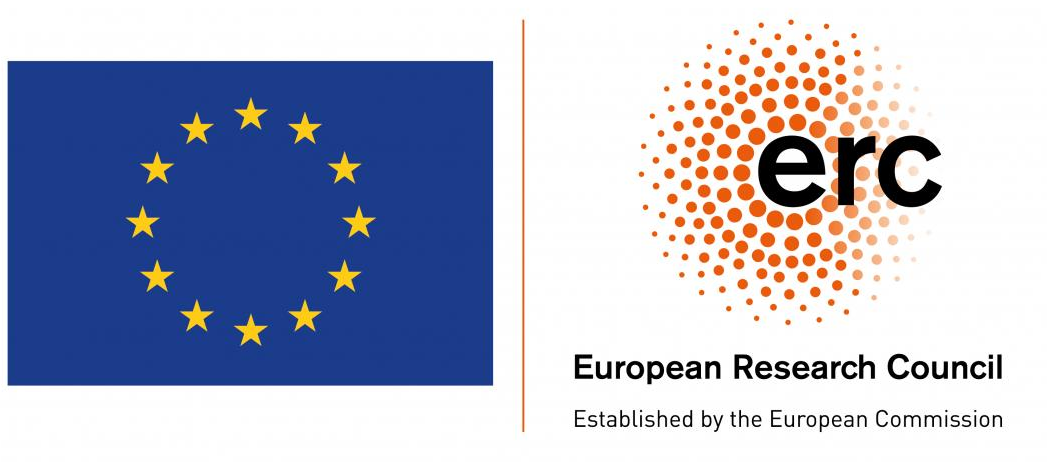}
    \end{center}
\end{acks}

\section*{Tool and computational resource disclosure}
During the preparation of this manuscript, the authors used ChatGPT as an interactive assistant. The tool was used for brainstorming and stress-testing ideas, exploring candidate potential functions, including during the process leading to the entropy potential studied in this paper, and checking algebraic derivations and small examples. All mathematical statements, proofs, examples, citations, and final text were checked, edited, and approved by the human authors, who take full responsibility for the content and attribution of the paper. 

\bibliographystyle{ACM-Reference-Format}
\bibliography{references}


\begin{thebibliography}{52}

%%% the Web address.

\ifx \showCODEN    \undefined \def \showCODEN     #1{\unskip}     \fi
\ifx \showDOI      \undefined \def \showDOI       #1{#1}\fi
\ifx \showISBNx    \undefined \def \showISBNx     #1{\unskip}     \fi
\ifx \showISBNxiii \undefined \def \showISBNxiii  #1{\unskip}     \fi
\ifx \showISSN     \undefined \def \showISSN      #1{\unskip}     \fi
\ifx \showLCCN     \undefined \def \showLCCN      #1{\unskip}     \fi
\ifx \shownote     \undefined \def \shownote      #1{#1}          \fi
\ifx \showarticletitle \undefined \def \showarticletitle #1{#1}   \fi
\ifx \showURL      \undefined \def \showURL       {\relax}        \fi
% invisible to TeX
\providecommand\bibfield[2]{#2}
\providecommand\bibinfo[2]{#2}
\providecommand\natexlab[1]{#1}
\providecommand\showeprint[2][]{arXiv:#2}

\bibitem[Agarwal et~al\mbox{.}(2021)]%
        {AgarwalEGISV2021}
\bibfield{author}{\bibinfo{person}{Aishwarya Agarwal}, \bibinfo{person}{Edith
  Elkind}, \bibinfo{person}{Jiarui Gan}, \bibinfo{person}{Ayumi Igarashi},
  \bibinfo{person}{Warut Suksompong}, {and} \bibinfo{person}{Alexandros~A.
  Voudouris}.} \bibinfo{year}{2021}\natexlab{}.
\newblock \showarticletitle{Schelling Games on Graphs}.
\newblock \bibinfo{journal}{\emph{\ifshortJournalName{}Artif.
  Intell.\else{}Artificial Intelligence\fi{}}}  \bibinfo{volume}{301}
  (\bibinfo{year}{2021}), \bibinfo{pages}{103576}.
\newblock
\urldef\tempurl%
\url{https://doi.org/10.1016/j.artint.2021.103576}
\showDOI{\tempurl}


\bibitem[Anshelevich et~al\mbox{.}(2008)]%
        {AnshelevichDKTWR2008}
\bibfield{author}{\bibinfo{person}{Elliot Anshelevich},
  \bibinfo{person}{Anirban Dasgupta}, \bibinfo{person}{Jon~M. Kleinberg},
  \bibinfo{person}{{\'{E}}va Tardos}, \bibinfo{person}{Tom Wexler}, {and}
  \bibinfo{person}{Tim Roughgarden}.} \bibinfo{year}{2008}\natexlab{}.
\newblock \showarticletitle{The Price of Stability for Network Design With Fair
  Cost Allocation}.
\newblock \bibinfo{journal}{\emph{\ifshortJournalName{}SIAM J.
  Comput.\else{}SIAM Journal on Computing\fi{}}} \bibinfo{volume}{38},
  \bibinfo{number}{4} (\bibinfo{year}{2008}), \bibinfo{pages}{1602--1623}.
\newblock
\urldef\tempurl%
\url{https://doi.org/10.1137/070680096}
\showDOI{\tempurl}


\bibitem[Aziz et~al\mbox{.}(2019)]%
        {AzizBBHOP2019}
\bibfield{author}{\bibinfo{person}{Haris Aziz}, \bibinfo{person}{Florian
  Brandl}, \bibinfo{person}{Felix Brandt}, \bibinfo{person}{Paul Harrenstein},
  \bibinfo{person}{Martin Olsen}, {and} \bibinfo{person}{Dominik Peters}.}
  \bibinfo{year}{2019}\natexlab{}.
\newblock \showarticletitle{Fractional Hedonic Games}.
\newblock \bibinfo{journal}{\emph{\ifshortJournalName{}{ACM} Trans. Econ.
  Comput.\else{}{ACM} Transactions on Economics and Computation\fi{}}}
  \bibinfo{volume}{7}, \bibinfo{number}{2}, Article \bibinfo{articleno}{6}
  (\bibinfo{year}{2019}), \bibinfo{numpages}{29}~pages.
\newblock
\urldef\tempurl%
\url{https://doi.org/10.1145/3327970}
\showDOI{\tempurl}


\bibitem[Bil{\`{o}} et~al\mbox{.}(2023)]%
        {BiloBDLMS2023}
\bibfield{author}{\bibinfo{person}{Davide Bil{\`{o}}},
  \bibinfo{person}{Vittorio Bil{\`{o}}}, \bibinfo{person}{Michelle
  D{\"{o}}ring}, \bibinfo{person}{Pascal Lenzner}, \bibinfo{person}{Louise
  Molitor}, {and} \bibinfo{person}{Jonas Schmidt}.}
  \bibinfo{year}{2023}\natexlab{}.
\newblock \showarticletitle{Schelling Games With Continuous Types}. In
  \bibinfo{booktitle}{\emph{\ifshortConferenceName{}Proc.
  {IJCAI}~'23\else{}Proceedings of the 32nd International Joint Conference on
  Artificial Intelligence, {IJCAI}~'23\fi{}}} (Macao, SAR, China),
  \bibfield{editor}{\bibinfo{person}{Edith Elkind}} (Ed.).
  \bibinfo{publisher}{ijcai.org}, \bibinfo{pages}{2520--2527}.
\newblock
\urldef\tempurl%
\url{https://doi.org/10.24963/ijcai.2023/280}
\showDOI{\tempurl}


\bibitem[Bil{\`{o}} et~al\mbox{.}(2022a)]%
        {BiloBLM2022a}
\bibfield{author}{\bibinfo{person}{Davide Bil{\`{o}}},
  \bibinfo{person}{Vittorio Bil{\`{o}}}, \bibinfo{person}{Pascal Lenzner},
  {and} \bibinfo{person}{Louise Molitor}.} \bibinfo{year}{2022}\natexlab{a}.
\newblock \showarticletitle{Tolerance Is Necessary for Stability: Single-Peaked
  Swap {S}chelling Games}. In
  \bibinfo{booktitle}{\emph{\ifshortConferenceName{}Proc.
  {IJCAI}~'22\else{}Proceedings of the 31st International Joint Conference on
  Artificial Intelligence, {IJCAI}~'22\fi{}}} (Vienna, Austria),
  \bibfield{editor}{\bibinfo{person}{Luc~De Raedt}} (Ed.).
  \bibinfo{publisher}{ijcai.org}, \bibinfo{pages}{81--87}.
\newblock
\urldef\tempurl%
\url{https://doi.org/10.24963/ijcai.2022/12}
\showDOI{\tempurl}


\bibitem[Bil{\`{o}} et~al\mbox{.}(2022b)]%
        {BiloBLM2022b}
\bibfield{author}{\bibinfo{person}{Davide Bil{\`{o}}},
  \bibinfo{person}{Vittorio Bil{\`{o}}}, \bibinfo{person}{Pascal Lenzner},
  {and} \bibinfo{person}{Louise Molitor}.} \bibinfo{year}{2022}\natexlab{b}.
\newblock \showarticletitle{Topological Influence and Locality in Swap
  {S}chelling Games}.
\newblock \bibinfo{journal}{\emph{\ifshortJournalName{}Auton. Agents Multi
  Agent Syst.\else{}Autonomous Agents and Multi-Agent Systems\fi{}}}
  \bibinfo{volume}{36}, \bibinfo{number}{2} (\bibinfo{year}{2022}),
  \bibinfo{pages}{47}.
\newblock
\urldef\tempurl%
\url{https://doi.org/10.1007/s10458-022-09573-7}
\showDOI{\tempurl}


\bibitem[Bil{\`o} et~al\mbox{.}(2018)]%
        {BiloFFMM2018}
\bibfield{author}{\bibinfo{person}{Vittorio Bil{\`o}}, \bibinfo{person}{Angelo
  Fanelli}, \bibinfo{person}{Michele Flammini}, \bibinfo{person}{Gianpiero
  Monaco}, {and} \bibinfo{person}{Luca Moscardelli}.}
  \bibinfo{year}{2018}\natexlab{}.
\newblock \showarticletitle{Nash Stable Outcomes in Fractional Hedonic Games:
  {E}xistence, Efficiency and Computation}.
\newblock \bibinfo{journal}{\emph{\ifshortJournalName{}J. Artif. Intell.
  Res.\else{}Journal of Artificial Intelligence Research\fi{}}}
  \bibinfo{volume}{62} (\bibinfo{year}{2018}), \bibinfo{pages}{315--371}.
\newblock
\urldef\tempurl%
\url{https://doi.org/10.1613/jair.1.11211}
\showDOI{\tempurl}


\bibitem[Bil{\`{o}} et~al\mbox{.}(2020)]%
        {BiloFM2020}
\bibfield{author}{\bibinfo{person}{Vittorio Bil{\`{o}}},
  \bibinfo{person}{Michele Flammini}, {and} \bibinfo{person}{Luca
  Moscardelli}.} \bibinfo{year}{2020}\natexlab{}.
\newblock \showarticletitle{The Price of Stability for Undirected Broadcast
  Network Design With Fair Cost Allocation Is Constant}.
\newblock \bibinfo{journal}{\emph{\ifshortJournalName{}Games Econ.
  Behav.\else{}Games and Economic Behavior\fi{}}}  \bibinfo{volume}{123}
  (\bibinfo{year}{2020}), \bibinfo{pages}{359--376}.
\newblock
\urldef\tempurl%
\url{https://doi.org/10.1016/j.geb.2014.09.010}
\showDOI{\tempurl}


\bibitem[Boehmer and Elkind(2020)]%
        {BoehmerE2020}
\bibfield{author}{\bibinfo{person}{Niclas Boehmer} {and} \bibinfo{person}{Edith
  Elkind}.} \bibinfo{year}{2020}\natexlab{}.
\newblock \showarticletitle{Individual-Based Stability in Hedonic Diversity
  Games}. In \bibinfo{booktitle}{\emph{\ifshortConferenceName{}Proc.
  {AAAI}~'20\else{}Proceedings of the 34th {AAAI} Conference on Artificial
  Intelligence, {AAAI}~'20\fi{}}} (New York, NY, USA),
  \bibfield{editor}{\bibinfo{person}{Vincent Conitzer} {and}
  \bibinfo{person}{Fei Sha}} (Eds.). \bibinfo{publisher}{AAAI Press},
  \bibinfo{address}{Palo Alto, CA, USA}, \bibinfo{pages}{1822--1829}.
\newblock
\urldef\tempurl%
\url{https://doi.org/10.1609/aaai.V34I02.5549}
\showDOI{\tempurl}


\bibitem[Bogomolnaia and Jackson(2002)]%
        {BogomolnaiaJ2002}
\bibfield{author}{\bibinfo{person}{Anna Bogomolnaia} {and}
  \bibinfo{person}{Matthew~O. Jackson}.} \bibinfo{year}{2002}\natexlab{}.
\newblock \showarticletitle{The Stability of Hedonic Coalition Structures}.
\newblock \bibinfo{journal}{\emph{\ifshortJournalName{}Games Econ.
  Behav.\else{}Games and Economic Behavior\fi{}}} \bibinfo{volume}{38},
  \bibinfo{number}{2} (\bibinfo{year}{2002}), \bibinfo{pages}{201--230}.
\newblock
\urldef\tempurl%
\url{https://doi.org/10.1006/game.2001.0877}
\showDOI{\tempurl}


\bibitem[Brandt et~al\mbox{.}(2012)]%
        {BrandtIKK2012}
\bibfield{author}{\bibinfo{person}{Christina Brandt}, \bibinfo{person}{Nicole
  Immorlica}, \bibinfo{person}{Gautam Kamath}, {and} \bibinfo{person}{Robert
  Kleinberg}.} \bibinfo{year}{2012}\natexlab{}.
\newblock \showarticletitle{An Analysis of One-Dimensional {S}chelling
  Segregation}. In \bibinfo{booktitle}{\emph{\ifshortConferenceName{}Proc.
  {STOC}~'12\else{}Proceedings of the 44th Symposium on Theory of Computing
  Conference, {STOC}~'12\fi{}}} (New York, NY, USA),
  \bibfield{editor}{\bibinfo{person}{Howard~J. Karloff} {and}
  \bibinfo{person}{Toniann Pitassi}} (Eds.). \bibinfo{publisher}{{ACM}},
  \bibinfo{address}{New York, NY, USA}, \bibinfo{pages}{789--804}.
\newblock
\urldef\tempurl%
\url{https://doi.org/10.1145/2213977.2214048}
\showDOI{\tempurl}


\bibitem[Bredereck et~al\mbox{.}(2019)]%
        {BredereckEI2019}
\bibfield{author}{\bibinfo{person}{Robert Bredereck}, \bibinfo{person}{Edith
  Elkind}, {and} \bibinfo{person}{Ayumi Igarashi}.}
  \bibinfo{year}{2019}\natexlab{}.
\newblock \showarticletitle{Hedonic Diversity Games}. In
  \bibinfo{booktitle}{\emph{\ifshortConferenceName{}Proc.
  {AAMAS}~'19\else{}Proceedings of the 18th International Conference on
  Autonomous Agents and Multiagent Systems, {AAMAS}~'19\fi{}}} (Montreal, QC,
  Canada), \bibfield{editor}{\bibinfo{person}{Edith Elkind},
  \bibinfo{person}{Manuela Veloso}, \bibinfo{person}{Noa Agmon}, {and}
  \bibinfo{person}{Matthew~E. Taylor}} (Eds.). \bibinfo{publisher}{IFAAMAS},
  \bibinfo{address}{Richland, SC, USA}, \bibinfo{pages}{565--573}.
\newblock
\urldef\tempurl%
\url{https://dl.acm.org/doi/10.5555/3306127.3331741}
\showURL{%
\tempurl}


\bibitem[Bullinger et~al\mbox{.}(2021)]%
        {BullingerSV2021}
\bibfield{author}{\bibinfo{person}{Martin Bullinger}, \bibinfo{person}{Warut
  Suksompong}, {and} \bibinfo{person}{Alexandros~A. Voudouris}.}
  \bibinfo{year}{2021}\natexlab{}.
\newblock \showarticletitle{Welfare Guarantees in {S}chelling Segregation}.
\newblock \bibinfo{journal}{\emph{\ifshortJournalName{}J. Artif. Intell.
  Res.\else{}Journal of Artificial Intelligence Research\fi{}}}
  \bibinfo{volume}{71} (\bibinfo{year}{2021}), \bibinfo{pages}{143--174}.
\newblock
\urldef\tempurl%
\url{https://doi.org/10.1613/jair.1.12771}
\showDOI{\tempurl}


\bibitem[Chauhan et~al\mbox{.}(2018)]%
        {ChauhanLM2018}
\bibfield{author}{\bibinfo{person}{Ankit Chauhan}, \bibinfo{person}{Pascal
  Lenzner}, {and} \bibinfo{person}{Louise Molitor}.}
  \bibinfo{year}{2018}\natexlab{}.
\newblock \showarticletitle{Schelling Segregation With Strategic Agents}. In
  \bibinfo{booktitle}{\emph{\ifshortConferenceName{}Proc.
  {SAGT}~'18\else{}Proceedings of the 11th International Symposium on
  Algorithmic Game Theory, {SAGT}~'18\fi{}}} (Beijing, China)
  \emph{(\bibinfo{series}{\ifshortConferenceName{}LNCS\else{}Lecture Notes in
  Computer Science\fi{}}, Vol.~\bibinfo{volume}{11059})},
  \bibfield{editor}{\bibinfo{person}{Xiaotie Deng}} (Ed.).
  \bibinfo{publisher}{Springer}, \bibinfo{address}{Cham},
  \bibinfo{pages}{137--149}.
\newblock
\urldef\tempurl%
\url{https://doi.org/10.1007/978-3-319-99660-8_13}
\showDOI{\tempurl}


\bibitem[de~La~Haye et~al\mbox{.}(2026)]%
        {continuousHLSW26}
\bibfield{author}{\bibinfo{person}{Merlin de La~Haye}, \bibinfo{person}{Pascal
  Lenzner}, \bibinfo{person}{Farehe Soheil}, {and} \bibinfo{person}{Marcus
  Wunderlich}.} \bibinfo{year}{2026}\natexlab{}.
\newblock \showarticletitle{Metric Hedonic Games on the Line}. In
  \bibinfo{booktitle}{\emph{\ifshortConferenceName{}Proc.
  {AAMAS}~'26\else{}Proceedings of the 25th International Conference on
  Autonomous Agents and Multiagent Systems, {AAMAS}~'26\fi{}}} (Paphos,
  Cyprus), \bibfield{editor}{\bibinfo{person}{Chris Amato},
  \bibinfo{person}{Louise Dennis}, \bibinfo{person}{Viviana Mascardi}, {and}
  \bibinfo{person}{John Thangarajah}} (Eds.). \bibinfo{publisher}{IFAAMAS},
  \bibinfo{address}{Richland, SC, USA}, \bibinfo{pages}{2214--2222}.
\newblock
\urldef\tempurl%
\url{https://doi.org/10.65109/CJCT2898}
\showDOI{\tempurl}


\bibitem[Deligkas et~al\mbox{.}(2024)]%
        {DeligkasEG2024}
\bibfield{author}{\bibinfo{person}{Argyrios Deligkas}, \bibinfo{person}{Eduard
  Eiben}, {and} \bibinfo{person}{Tiger{-}Lily Goldsmith}.}
  \bibinfo{year}{2024}\natexlab{}.
\newblock \showarticletitle{The Parameterized Complexity of Welfare Guarantees
  in {S}chelling Segregation}.
\newblock \bibinfo{journal}{\emph{\ifshortJournalName{}Theor. Comput.
  Sci.\else{}Theoretical Computer Science\fi{}}}  \bibinfo{volume}{1017}
  (\bibinfo{year}{2024}), \bibinfo{pages}{114783}.
\newblock
\urldef\tempurl%
\url{https://doi.org/10.1016/j.tcs.2024.114783}
\showDOI{\tempurl}


\bibitem[Doerr et~al\mbox{.}(2012)]%
        {DoerrJW12}
\bibfield{author}{\bibinfo{person}{Benjamin Doerr}, \bibinfo{person}{Daniel
  Johannsen}, {and} \bibinfo{person}{Carola Winzen}.}
  \bibinfo{year}{2012}\natexlab{}.
\newblock \showarticletitle{Multiplicative Drift Analysis}.
\newblock
  \bibinfo{journal}{\emph{\ifshortJournalName{}Algorithmica\else{}Algorithmica\fi{}}}
  \bibinfo{volume}{64}, \bibinfo{number}{4} (\bibinfo{year}{2012}),
  \bibinfo{pages}{673--697}.
\newblock
\urldef\tempurl%
\url{https://doi.org/10.1007/s00453-012-9622-x}
\showDOI{\tempurl}


\bibitem[Dreze and Greenberg(1980)]%
        {DrezeG1980}
\bibfield{author}{\bibinfo{person}{Jacques~H. Dreze} {and}
  \bibinfo{person}{Joseph Greenberg}.} \bibinfo{year}{1980}\natexlab{}.
\newblock \showarticletitle{Hedonic Coalitions: {O}ptimality and Stability}.
\newblock
  \bibinfo{journal}{\emph{\ifshortJournalName{}Econometrica\else{}Econometrica\fi{}}}
  \bibinfo{volume}{48}, \bibinfo{number}{4} (\bibinfo{year}{1980}),
  \bibinfo{pages}{987--1003}.
\newblock
\urldef\tempurl%
\url{https://doi.org/10.2307/1912943}
\showDOI{\tempurl}


\bibitem[Echzell et~al\mbox{.}(2019)]%
        {Echzell0LMPSSS2019}
\bibfield{author}{\bibinfo{person}{Hagen Echzell}, \bibinfo{person}{Tobias
  Friedrich}, \bibinfo{person}{Pascal Lenzner}, \bibinfo{person}{Louise
  Molitor}, \bibinfo{person}{Marcus Pappik}, \bibinfo{person}{Friedrich
  Sch{\"{o}}ne}, \bibinfo{person}{Fabian Sommer}, {and} \bibinfo{person}{David
  Stangl}.} \bibinfo{year}{2019}\natexlab{}.
\newblock \showarticletitle{Convergence and Hardness of Strategic {S}chelling
  Segregation}. In \bibinfo{booktitle}{\emph{\ifshortConferenceName{}Proc.
  WINE~'19\else{}Proceedings of the 15th International Conference on Web and
  Internet Economics, {WINE}~'19\fi{}}} (New York, NY, USA)
  \emph{(\bibinfo{series}{\ifshortConferenceName{}LNCS\else{}Lecture Notes in
  Computer Science\fi{}}, Vol.~\bibinfo{volume}{11920})},
  \bibfield{editor}{\bibinfo{person}{Ioannis Caragiannis},
  \bibinfo{person}{Vahab~S. Mirrokni}, {and} \bibinfo{person}{Evdokia
  Nikolova}} (Eds.). \bibinfo{publisher}{Springer}, \bibinfo{address}{Cham},
  \bibinfo{pages}{156--170}.
\newblock
\urldef\tempurl%
\url{https://doi.org/10.1007/978-3-030-35389-6_12}
\showDOI{\tempurl}


\bibitem[Even{-}Dar et~al\mbox{.}(2007)]%
        {EvenDarKM2007}
\bibfield{author}{\bibinfo{person}{Eyal Even{-}Dar}, \bibinfo{person}{Alexander
  Kesselman}, {and} \bibinfo{person}{Yishay Mansour}.}
  \bibinfo{year}{2007}\natexlab{}.
\newblock \showarticletitle{Convergence Time to {N}ash Equilibrium in Load
  Balancing}.
\newblock \bibinfo{journal}{\emph{\ifshortJournalName{}{ACM} Trans.
  Algorithms\else{}{ACM} Transactions on Algorithms\fi{}}} \bibinfo{volume}{3},
  \bibinfo{number}{3} (\bibinfo{year}{2007}), \bibinfo{pages}{32}.
\newblock
\urldef\tempurl%
\url{https://doi.org/10.1145/1273340.1273348}
\showDOI{\tempurl}


\bibitem[Fabrikant et~al\mbox{.}(2004)]%
        {FabrikantPT2004}
\bibfield{author}{\bibinfo{person}{Alex Fabrikant}, \bibinfo{person}{Christos
  Papadimitriou}, {and} \bibinfo{person}{Kunal Talwar}.}
  \bibinfo{year}{2004}\natexlab{}.
\newblock \showarticletitle{The Complexity of Pure {N}ash Equilibria}. In
  \bibinfo{booktitle}{\emph{\ifshortConferenceName{}Proc.
  {STOC}~'04\else{}Proceedings of the 36th Annual {ACM} Symposium on Theory of
  Computing, {STOC}~'04\fi{}}} (Chicago, IL, USA),
  \bibfield{editor}{\bibinfo{person}{László Babai}} (Ed.).
  \bibinfo{publisher}{ACM}, \bibinfo{address}{New York, NY, USA},
  \bibinfo{pages}{604--612}.
\newblock
\urldef\tempurl%
\url{https://doi.org/10.1145/1007352.1007445}
\showDOI{\tempurl}


\bibitem[Feldotto et~al\mbox{.}(2019)]%
        {FeldottoLMS2019}
\bibfield{author}{\bibinfo{person}{Matthias Feldotto}, \bibinfo{person}{Pascal
  Lenzner}, \bibinfo{person}{Louise Molitor}, {and} \bibinfo{person}{Alexander
  Skopalik}.} \bibinfo{year}{2019}\natexlab{}.
\newblock \showarticletitle{From Hotelling to Load Balancing: Approximation and
  the Principle of Minimum Differentiation}. In
  \bibinfo{booktitle}{\emph{\ifshortConferenceName{}Proc.
  {AAMAS}~'19\else{}Proceedings of the 18th International Conference on
  Autonomous Agents and Multiagent Systems, {AAMAS}~'19\fi{}}} (Montreal, QC,
  Canada), \bibfield{editor}{\bibinfo{person}{Edith Elkind},
  \bibinfo{person}{Manuela Veloso}, \bibinfo{person}{Noa Agmon}, {and}
  \bibinfo{person}{Matthew~E. Taylor}} (Eds.). \bibinfo{publisher}{IFAAMAS},
  \bibinfo{address}{Richland, SC, USA}, \bibinfo{pages}{1949--1951}.
\newblock
\urldef\tempurl%
\url{https://dl.acm.org/doi/10.5555/3306127.3331973}
\showURL{%
\tempurl}


\bibitem[Fredman and Tarjan(1987)]%
        {FredmanT1987}
\bibfield{author}{\bibinfo{person}{Michael~L. Fredman} {and}
  \bibinfo{person}{Robert~Endre Tarjan}.} \bibinfo{year}{1987}\natexlab{}.
\newblock \showarticletitle{Fibonacci Heaps and Their Uses in Improved Network
  Optimization Algorithms}.
\newblock \bibinfo{journal}{\emph{\ifshortJournalName{}J.~{ACM}\else{}Journal
  of the {ACM}\fi{}}} \bibinfo{volume}{34}, \bibinfo{number}{3}
  (\bibinfo{year}{1987}), \bibinfo{pages}{596--615}.
\newblock
\urldef\tempurl%
\url{https://doi.org/10.1145/28869.28874}
\showDOI{\tempurl}


\bibitem[Friedrich et~al\mbox{.}(2023)]%
        {singlepeakedFriedrichLMS23}
\bibfield{author}{\bibinfo{person}{Tobias Friedrich}, \bibinfo{person}{Pascal
  Lenzner}, \bibinfo{person}{Louise Molitor}, {and} \bibinfo{person}{Lars
  Seifert}.} \bibinfo{year}{2023}\natexlab{}.
\newblock \showarticletitle{Single-Peaked Jump Schelling Games}. In
  \bibinfo{booktitle}{\emph{\ifshortConferenceName{}Proc.
  {SAGT}~'23\else{}Proceedings of the 16th International Symposium on
  Algorithmic Game Theory, {SAGT}~'23\fi{}}} (Egham, United Kingdom)
  \emph{(\bibinfo{series}{\ifshortConferenceName{}LNCS\else{}Lecture Notes in
  Computer Science\fi{}}, Vol.~\bibinfo{volume}{14238})},
  \bibfield{editor}{\bibinfo{person}{Argyrios Deligkas} {and}
  \bibinfo{person}{Aris Filos{-}Ratsikas}} (Eds.).
  \bibinfo{publisher}{Springer}, \bibinfo{address}{Cham},
  \bibinfo{pages}{111--126}.
\newblock
\urldef\tempurl%
\url{https://doi.org/10.1007/978-3-031-43254-5_7}
\showDOI{\tempurl}


\bibitem[{Gadea Harder} et~al\mbox{.}(2023)]%
        {GadeaHarderKLS2023}
\bibfield{author}{\bibinfo{person}{Jonathan {Gadea Harder}},
  \bibinfo{person}{Simon Krogmann}, \bibinfo{person}{Pascal Lenzner}, {and}
  \bibinfo{person}{Alexander Skopalik}.} \bibinfo{year}{2023}\natexlab{}.
\newblock \showarticletitle{Strategic Resource Selection With Homophilic
  Agents}. In \bibinfo{booktitle}{\emph{\ifshortConferenceName{}Proc.
  {IJCAI}~'23\else{}Proceedings of the 32nd International Joint Conference on
  Artificial Intelligence, {IJCAI}~'23\fi{}}} (Macao, SAR, China),
  \bibfield{editor}{\bibinfo{person}{Edith Elkind}} (Ed.).
  \bibinfo{publisher}{ijcai.org}, \bibinfo{pages}{2701--2709}.
\newblock
\urldef\tempurl%
\url{https://doi.org/10.24963/ijcai.2023/301}
\showDOI{\tempurl}


\bibitem[Ganian et~al\mbox{.}(2023)]%
        {GanianHKSS2023}
\bibfield{author}{\bibinfo{person}{Robert Ganian}, \bibinfo{person}{Thekla
  Hamm}, \bibinfo{person}{Dušan Knop}, \bibinfo{person}{Šimon Schierreich},
  {and} \bibinfo{person}{Ondřej Such{\'{y}}}.}
  \bibinfo{year}{2023}\natexlab{}.
\newblock \showarticletitle{Hedonic Diversity Games: {A} Complexity Picture
  With More Than Two Colors}.
\newblock \bibinfo{journal}{\emph{\ifshortJournalName{}Artif.
  Intell.\else{}Artificial Intelligence\fi{}}}  \bibinfo{volume}{325}
  (\bibinfo{year}{2023}), \bibinfo{pages}{104017}.
\newblock
\urldef\tempurl%
\url{https://doi.org/10.1016/j.artint.2023.104017}
\showDOI{\tempurl}


\bibitem[Hajek(1982)]%
        {Hajek1982}
\bibfield{author}{\bibinfo{person}{Bruce Hajek}.}
  \bibinfo{year}{1982}\natexlab{}.
\newblock \showarticletitle{Hitting-Time and Occupation-Time Bounds Implied by
  Drift Analysis With Applications}.
\newblock \bibinfo{journal}{\emph{\ifshortJournalName{}Adv. Appl.
  Probab.\else{}Advances in Applied Probability\fi{}}} \bibinfo{volume}{14},
  \bibinfo{number}{3} (\bibinfo{year}{1982}), \bibinfo{pages}{502--525}.
\newblock
\urldef\tempurl%
\url{https://doi.org/10.2307/1426671}
\showDOI{\tempurl}


\bibitem[Harks et~al\mbox{.}(2011)]%
        {HarksKM2011}
\bibfield{author}{\bibinfo{person}{Tobias Harks}, \bibinfo{person}{Max Klimm},
  {and} \bibinfo{person}{Rolf~H. M{\"{o}}hring}.}
  \bibinfo{year}{2011}\natexlab{}.
\newblock \showarticletitle{Characterizing the Existence of Potential Functions
  in Weighted Congestion Games}.
\newblock \bibinfo{journal}{\emph{\ifshortJournalName{}Theory Comput.
  Syst.\else{}Theory of Computing Systems\fi{}}} \bibinfo{volume}{49},
  \bibinfo{number}{1} (\bibinfo{year}{2011}), \bibinfo{pages}{46--70}.
\newblock
\urldef\tempurl%
\url{https://doi.org/10.1007/S00224-011-9315-X}
\showDOI{\tempurl}


\bibitem[He and Yao(2001)]%
        {HeY2001}
\bibfield{author}{\bibinfo{person}{Jun He} {and} \bibinfo{person}{Xin Yao}.}
  \bibinfo{year}{2001}\natexlab{}.
\newblock \showarticletitle{Drift Analysis and Average Time Complexity of
  Evolutionary Algorithms}.
\newblock \bibinfo{journal}{\emph{\ifshortJournalName{}Artif.
  Intell.\else{}Artificial Intelligence\fi{}}} \bibinfo{volume}{127},
  \bibinfo{number}{1} (\bibinfo{year}{2001}), \bibinfo{pages}{57--85}.
\newblock
\urldef\tempurl%
\url{https://doi.org/10.1016/s0004-3702(01)00058-3}
\showDOI{\tempurl}


\bibitem[Kanellopoulos et~al\mbox{.}(2023)]%
        {KanellopoulosKV2023}
\bibfield{author}{\bibinfo{person}{Panagiotis Kanellopoulos},
  \bibinfo{person}{Maria Kyropoulou}, {and} \bibinfo{person}{Alexandros~A.
  Voudouris}.} \bibinfo{year}{2023}\natexlab{}.
\newblock \showarticletitle{Not All Strangers Are the Same: The Impact of
  Tolerance in {S}chelling Games}.
\newblock \bibinfo{journal}{\emph{\ifshortJournalName{}Theor. Comput.
  Sci.\else{}Theoretical Computer Science\fi{}}}  \bibinfo{volume}{971}
  (\bibinfo{year}{2023}), \bibinfo{pages}{114065}.
\newblock
\urldef\tempurl%
\url{https://doi.org/10.1016/j.tcs.2023.114065}
\showDOI{\tempurl}


\bibitem[Knop and Šimon Schierreich(2023)]%
        {KnopS2023}
\bibfield{author}{\bibinfo{person}{Dušan Knop} {and} \bibinfo{person}{Šimon
  Schierreich}.} \bibinfo{year}{2023}\natexlab{}.
\newblock \showarticletitle{Host Community Respecting Refugee Housing}. In
  \bibinfo{booktitle}{\emph{\ifshortConferenceName{}Proc.
  {AAMAS}~'23\else{}Proceedings of the 22nd International Conference on
  Autonomous Agents and Multiagent Systems, {AAMAS}~'23\fi{}}} (London, United
  Kingdom), \bibfield{editor}{\bibinfo{person}{Noa Agmon},
  \bibinfo{person}{Bo~An}, \bibinfo{person}{Alessandro Ricci}, {and}
  \bibinfo{person}{William Yeoh}} (Eds.). \bibinfo{publisher}{IFAAMAS},
  \bibinfo{address}{Richland, SC, USA}, \bibinfo{pages}{966--975}.
\newblock
\urldef\tempurl%
\url{https://dl.acm.org/doi/10.5555/3545946.3598736}
\showURL{%
\tempurl}


\bibitem[Kreisel et~al\mbox{.}(2024)]%
        {KreiselBFN2024}
\bibfield{author}{\bibinfo{person}{Luca Kreisel}, \bibinfo{person}{Niclas
  Boehmer}, \bibinfo{person}{Vincent Froese}, {and} \bibinfo{person}{Rolf
  Niedermeier}.} \bibinfo{year}{2024}\natexlab{}.
\newblock \showarticletitle{Equilibria in {S}chelling Games: {C}omputational
  Hardness and Robustness}.
\newblock \bibinfo{journal}{\emph{\ifshortJournalName{}Auton. Agents Multi
  Agent Syst.\else{}Autonomous Agents and Multi-Agent Systems\fi{}}}
  \bibinfo{volume}{38}, \bibinfo{number}{1} (\bibinfo{year}{2024}),
  \bibinfo{pages}{9}.
\newblock
\urldef\tempurl%
\url{https://doi.org/10.1007/s10458-023-09632-7}
\showDOI{\tempurl}


\bibitem[Krogmann et~al\mbox{.}(2021)]%
        {twostageKrogmannLMS21}
\bibfield{author}{\bibinfo{person}{Simon Krogmann}, \bibinfo{person}{Pascal
  Lenzner}, \bibinfo{person}{Louise Molitor}, {and} \bibinfo{person}{Alexander
  Skopalik}.} \bibinfo{year}{2021}\natexlab{}.
\newblock \showarticletitle{Two-Stage Facility Location Games With Strategic
  Clients and Facilities}. In
  \bibinfo{booktitle}{\emph{\ifshortConferenceName{}Proc.
  {IJCAI}~'21\else{}Proceedings of the 30th International Joint Conference on
  Artificial Intelligence, {IJCAI}~'21\fi{}}},
  \bibfield{editor}{\bibinfo{person}{Zhi{-}Hua Zhou}} (Ed.).
  \bibinfo{publisher}{ijcai.org}, \bibinfo{pages}{292--298}.
\newblock
\urldef\tempurl%
\url{https://doi.org/10.24963/IJCAI.2021/41}
\showDOI{\tempurl}


\bibitem[Krogmann et~al\mbox{.}(2023)]%
        {twostageKrogmannLS23}
\bibfield{author}{\bibinfo{person}{Simon Krogmann}, \bibinfo{person}{Pascal
  Lenzner}, {and} \bibinfo{person}{Alexander Skopalik}.}
  \bibinfo{year}{2023}\natexlab{}.
\newblock \showarticletitle{Strategic Facility Location With Clients That
  Minimize Total Waiting Time}. In
  \bibinfo{booktitle}{\emph{\ifshortConferenceName{}Proc.
  {AAAI}~'23\else{}Proceedings of the 37th {AAAI} Conference on Artificial
  Intelligence, {AAAI}~'23\fi{}}} (Washington, DC, USA),
  \bibfield{editor}{\bibinfo{person}{Brian Williams}, \bibinfo{person}{Yiling
  Chen}, {and} \bibinfo{person}{Jennifer Neville}} (Eds.).
  \bibinfo{publisher}{AAAI Press}, \bibinfo{address}{Washington, DC, USA},
  \bibinfo{pages}{5714--5721}.
\newblock
\urldef\tempurl%
\url{https://doi.org/10.1609/aaai.v37i5.25709}
\showDOI{\tempurl}


\bibitem[Krogmann et~al\mbox{.}(2025)]%
        {KrogmannLS2025}
\bibfield{author}{\bibinfo{person}{Simon Krogmann}, \bibinfo{person}{Pascal
  Lenzner}, {and} \bibinfo{person}{Alexander Skopalik}.}
  \bibinfo{year}{2025}\natexlab{}.
\newblock \showarticletitle{The Bakers and Millers Game With Restricted
  Locations}. In \bibinfo{booktitle}{\emph{\ifshortConferenceName{}Proc.
  {AAMAS}~'25\else{}Proceedings of the 24th International Conference on
  Autonomous Agents and Multiagent Systems, {AAMAS}~'25\fi{}}} (Detroit, MI,
  USA), \bibfield{editor}{\bibinfo{person}{Sanmay Das}, \bibinfo{person}{Ann
  Now{\'{e}}}, {and} \bibinfo{person}{Yevgeniy Vorobeychik}} (Eds.).
  \bibinfo{publisher}{IFAAMAS}, \bibinfo{address}{Richland, SC, USA},
  \bibinfo{pages}{1209--1217}.
\newblock
\urldef\tempurl%
\url{https://dl.acm.org/doi/10.5555/3709347.3743752}
\showURL{%
\tempurl}


\bibitem[Krogmann et~al\mbox{.}(2024)]%
        {twostageKrogmannLSUV24}
\bibfield{author}{\bibinfo{person}{Simon Krogmann}, \bibinfo{person}{Pascal
  Lenzner}, \bibinfo{person}{Alexander Skopalik}, \bibinfo{person}{Marc Uetz},
  {and} \bibinfo{person}{Marnix~C. Vos}.} \bibinfo{year}{2024}\natexlab{}.
\newblock \showarticletitle{Equilibria in Two-Stage Facility Location With
  Atomic Clients}. In \bibinfo{booktitle}{\emph{\ifshortConferenceName{}Proc.
  {IJCAI}~'24\else{}Proceedings of the 33rd International Joint Conference on
  Artificial Intelligence, {IJCAI}~'24\fi{}}} (Jeju, South Korea),
  \bibfield{editor}{\bibinfo{person}{Kate Larson}} (Ed.).
  \bibinfo{publisher}{ijcai.org}, \bibinfo{pages}{2842--2850}.
\newblock
\urldef\tempurl%
\url{https://doi.org/10.24963/ijcai.2024/315}
\showDOI{\tempurl}


\bibitem[Lengler(2020)]%
        {Lengler2020}
\bibfield{author}{\bibinfo{person}{Johannes Lengler}.}
  \bibinfo{year}{2020}\natexlab{}.
\newblock \showarticletitle{Drift Analysis}.
\newblock In \bibinfo{booktitle}{\emph{Theory of Evolutionary Computation:
  {R}ecent Developments in Discrete Optimization}},
  \bibfield{editor}{\bibinfo{person}{Benjamin Doerr} {and}
  \bibinfo{person}{Frank Neumann}} (Eds.). \bibinfo{publisher}{Springer},
  \bibinfo{address}{Cham}.
\newblock
\urldef\tempurl%
\url{https://doi.org/10.1007/978-3-030-29414-4}
\showDOI{\tempurl}


\bibitem[Milchtaich(1996)]%
        {MILCHTAICH1996111}
\bibfield{author}{\bibinfo{person}{Igal Milchtaich}.}
  \bibinfo{year}{1996}\natexlab{}.
\newblock \showarticletitle{Congestion Games With Player-Specific Payoff
  Functions}.
\newblock \bibinfo{journal}{\emph{\ifshortJournalName{}Games Econ.
  Behav.\else{}Games and Economic Behavior\fi{}}} \bibinfo{volume}{13},
  \bibinfo{number}{1} (\bibinfo{year}{1996}), \bibinfo{pages}{111--124}.
\newblock
\urldef\tempurl%
\url{https://doi.org/10.1006/game.1996.0027}
\showDOI{\tempurl}


\bibitem[Monderer and Shapley(1996)]%
        {MondererS1996}
\bibfield{author}{\bibinfo{person}{Dov Monderer} {and}
  \bibinfo{person}{Lloyd~S. Shapley}.} \bibinfo{year}{1996}\natexlab{}.
\newblock \showarticletitle{Potential Games}.
\newblock \bibinfo{journal}{\emph{\ifshortJournalName{}Games Econ.
  Behav.\else{}Games and Economic Behavior\fi{}}} \bibinfo{volume}{14},
  \bibinfo{number}{1} (\bibinfo{year}{1996}), \bibinfo{pages}{124--143}.
\newblock
\urldef\tempurl%
\url{https://doi.org/10.1006/game.1996.0044}
\showDOI{\tempurl}


\bibitem[Narayanan et~al\mbox{.}(2025a)]%
        {NarayananOTV2025}
\bibfield{author}{\bibinfo{person}{Lata Narayanan}, \bibinfo{person}{Jaroslav
  Opatrny}, \bibinfo{person}{Shanmukha Tummala}, {and}
  \bibinfo{person}{Alexandros~A. Voudouris}.} \bibinfo{year}{2025}\natexlab{a}.
\newblock \showarticletitle{Variety-Seeking Jump Games on Graphs}. In
  \bibinfo{booktitle}{\emph{\ifshortConferenceName{}Proc.
  {IJCAI}~'25\else{}Proceedings of the 34th International Joint Conference on
  Artificial Intelligence, {IJCAI}~'25\fi}} (Montreal, QC, Canada),
  \bibfield{editor}{\bibinfo{person}{James Kwok}} (Ed.).
  \bibinfo{publisher}{ijcai.org}, \bibinfo{pages}{4005--4013}.
\newblock
\urldef\tempurl%
\url{https://doi.org/10.24963/ijcai.2025/446}
\showDOI{\tempurl}


\bibitem[Narayanan et~al\mbox{.}(2025b)]%
        {NarayananSV2025}
\bibfield{author}{\bibinfo{person}{Lata Narayanan}, \bibinfo{person}{Yasaman
  Sabbagh}, {and} \bibinfo{person}{Alexandros~A. Voudouris}.}
  \bibinfo{year}{2025}\natexlab{b}.
\newblock \showarticletitle{Diversity-Seeking Jump Games in Networks}.
\newblock \bibinfo{journal}{\emph{\ifshortJournalName{}Auton. Agents Multi
  Agent Syst.\else{}Autonomous Agents and Multi-Agent Systems\fi{}}}
  \bibinfo{volume}{39}, \bibinfo{number}{2} (\bibinfo{year}{2025}),
  \bibinfo{pages}{32}.
\newblock
\urldef\tempurl%
\url{https://doi.org/10.1007/s10458-025-09714-8}
\showDOI{\tempurl}


\bibitem[Nash~Jr(1950)]%
        {nash1950equilibrium}
\bibfield{author}{\bibinfo{person}{John~F. Nash~Jr}.}
  \bibinfo{year}{1950}\natexlab{}.
\newblock \showarticletitle{Equilibrium Points in $N$-Person Games}.
\newblock \bibinfo{journal}{\emph{PNAS}} \bibinfo{volume}{36},
  \bibinfo{number}{1} (\bibinfo{year}{1950}), \bibinfo{pages}{48--49}.
\newblock


\bibitem[Rosenthal(1973)]%
        {Rosenthal1973}
\bibfield{author}{\bibinfo{person}{Robert~W. Rosenthal}.}
  \bibinfo{year}{1973}\natexlab{}.
\newblock \showarticletitle{A Class of Games Possessing Pure-Strategy {N}ash
  Equilibria}.
\newblock \bibinfo{journal}{\emph{\ifshortJournalName{}Int. J. Game
  Theory\else{}International Journal of Game Theory\fi{}}}  \bibinfo{volume}{2}
  (\bibinfo{year}{1973}), \bibinfo{pages}{65--67}.
\newblock
\urldef\tempurl%
\url{https://doi.org/10.1007/BF01737559}
\showDOI{\tempurl}


\bibitem[Roughgarden and Tardos(2002)]%
        {RoughgardenT2002}
\bibfield{author}{\bibinfo{person}{Tim Roughgarden} {and}
  \bibinfo{person}{{\'{E}}va Tardos}.} \bibinfo{year}{2002}\natexlab{}.
\newblock \showarticletitle{How Bad Is Selfish Routing?}
\newblock \bibinfo{journal}{\emph{\ifshortJournalName{}J.~{ACM}\else{}Journal
  of the {ACM}\fi{}}} \bibinfo{volume}{49}, \bibinfo{number}{2}
  (\bibinfo{year}{2002}), \bibinfo{pages}{236--259}.
\newblock
\urldef\tempurl%
\url{https://doi.org/10.1145/506147.506153}
\showDOI{\tempurl}


\bibitem[Schelling(1969)]%
        {Schelling1969}
\bibfield{author}{\bibinfo{person}{Thomas~C. Schelling}.}
  \bibinfo{year}{1969}\natexlab{}.
\newblock \showarticletitle{Models of Segregation}.
\newblock \bibinfo{journal}{\emph{The American Economic Review}}
  \bibinfo{volume}{59}, \bibinfo{number}{2} (\bibinfo{year}{1969}),
  \bibinfo{pages}{488--493}.
\newblock
\urldef\tempurl%
\url{http://www.jstor.org/stable/1823701}
\showURL{%
\tempurl}


\bibitem[Schelling(1971)]%
        {Schelling1971}
\bibfield{author}{\bibinfo{person}{Thomas~C. Schelling}.}
  \bibinfo{year}{1971}\natexlab{}.
\newblock \showarticletitle{Dynamic Models of Segregation}.
\newblock \bibinfo{journal}{\emph{The Journal of Mathematical Sociology}}
  \bibinfo{volume}{1}, \bibinfo{number}{2} (\bibinfo{year}{1971}),
  \bibinfo{pages}{143--186}.
\newblock
\urldef\tempurl%
\url{https://doi.org/10.1080/0022250X.1971.9989794}
\showDOI{\tempurl}


\bibitem[Schierreich(2023)]%
        {Schierreich2023}
\bibfield{author}{\bibinfo{person}{{\v{S}}imon Schierreich}.}
  \bibinfo{year}{2023}\natexlab{}.
\newblock \showarticletitle{Anonymous Refugee Housing With Upper-Bounds}.
\newblock \bibinfo{journal}{\emph{CoRR}}  \bibinfo{volume}{abs/2308.09501}
  (\bibinfo{year}{2023}).
\newblock
\urldef\tempurl%
\url{https://doi.org/10.48550/arxiv.2308.09501}
\showDOI{\tempurl}
\showeprint[arXiv]{2308.09501}


\bibitem[Schierreich(2024)]%
        {Schierreich2024}
\bibfield{author}{\bibinfo{person}{{\v{S}}imon Schierreich}.}
  \bibinfo{year}{2024}\natexlab{}.
\newblock \showarticletitle{Two-Stage Refugee Resettlement Models:
  {C}omputational Aspects of the Second Stage}. In
  \bibinfo{booktitle}{\emph{\ifshortConferenceName{}Proc.
  {AIES}~'24\else{}Proceedings of the 7th {AAAI/ACM} Conference on AI, Ethics,
  and Society, {AIES}~'24\fi{}}} (San Jose, CA, {USA}),
  \bibfield{editor}{\bibinfo{person}{Sanmay Das} {and}
  \bibinfo{person}{Brian~Patrick Green}} (Eds.). \bibinfo{publisher}{{AAAI}
  Press}, \bibinfo{address}{Washington, DC, USA}, \bibinfo{pages}{50--51}.
\newblock
\urldef\tempurl%
\url{https://doi.org/10.1609/aies.v7i2.31908}
\showDOI{\tempurl}


\bibitem[Sleator and Tarjan(1985)]%
        {SleatorT1985}
\bibfield{author}{\bibinfo{person}{Daniel~Dominic Sleator} {and}
  \bibinfo{person}{Robert~Endre Tarjan}.} \bibinfo{year}{1985}\natexlab{}.
\newblock \showarticletitle{Self-Adjusting Binary Search Trees}.
\newblock \bibinfo{journal}{\emph{\ifshortJournalName{}J.~{ACM}\else{}Journal
  of the {ACM}\fi{}}} \bibinfo{volume}{32}, \bibinfo{number}{3}
  (\bibinfo{year}{1985}), \bibinfo{pages}{652--686}.
\newblock
\urldef\tempurl%
\url{https://doi.org/10.1145/3828.3835}
\showDOI{\tempurl}


\bibitem[Tarjan(1975)]%
        {Tarjan1975}
\bibfield{author}{\bibinfo{person}{Robert~Endre Tarjan}.}
  \bibinfo{year}{1975}\natexlab{}.
\newblock \showarticletitle{Efficiency of a Good but Not Linear Set Union
  Algorithm}.
\newblock \bibinfo{journal}{\emph{\ifshortJournalName{}J.~{ACM}\else{}Journal
  of the {ACM}\fi{}}} \bibinfo{volume}{22}, \bibinfo{number}{2}
  (\bibinfo{year}{1975}), \bibinfo{pages}{215--225}.
\newblock
\urldef\tempurl%
\url{https://doi.org/10.1145/321879.321884}
\showDOI{\tempurl}


\bibitem[Tarjan(1985)]%
        {Tarjan1985}
\bibfield{author}{\bibinfo{person}{Robert~Endre Tarjan}.}
  \bibinfo{year}{1985}\natexlab{}.
\newblock \showarticletitle{Amortized Computational Complexity}.
\newblock \bibinfo{journal}{\emph{\ifshortJournalName{}SIAM J. Algebr. Discrete
  Methods\else{}SIAM Journal on Algebraic and Discrete Methods\fi{}}}
  \bibinfo{volume}{6}, \bibinfo{number}{2} (\bibinfo{year}{1985}),
  \bibinfo{pages}{306--318}.
\newblock
\urldef\tempurl%
\url{https://doi.org/10.1137/0606031}
\showDOI{\tempurl}


\bibitem[Witt(2013)]%
        {Witt13}
\bibfield{author}{\bibinfo{person}{Carsten Witt}.}
  \bibinfo{year}{2013}\natexlab{}.
\newblock \showarticletitle{Tight Bounds on the Optimization Time of a
  Randomized Search Heuristic on Linear Functions}.
\newblock \bibinfo{journal}{\emph{Combinatorics, Probability \& Computing}}
  \bibinfo{volume}{22}, \bibinfo{number}{2} (\bibinfo{year}{2013}),
  \bibinfo{pages}{294--318}.
\newblock
\urldef\tempurl%
\url{https://doi.org/10.1017/s0963548312000600}
\showDOI{\tempurl}


\end{thebibliography}

\end{document}